\documentclass[runningheads]{llncs}
\usepackage{graphicx}
\usepackage{longtable}
\usepackage{multirow}

\begin{document}
\title{Two-Bar Charts Packing Problem\thanks{The work is supported by Mathematical Center in Akademgorodok under agreement No 075-2019-1613 with the Ministry of Science and Higher Education of the Russian Federation.}}
\author{Adil Erzin\orcidID{0000-0002-2183-523X} \and Gregory Melidi \and Stepan
Nazarenko \and Roman Plotnikov\orcidID{0000-0003-2038-5609}}
\authorrunning{A. Erzin et al.}
\institute{Sobolev Institute of Mathematics, SB RAS, Novosibirsk 630090, Russia\\
\email{adilerzin@math.nsc.ru}}

\maketitle              
\begin{abstract}
We consider a Bar Charts Packing Problem (BCPP), in which it is necessary to pack bar charts (BCs) in a strip of minimum length. The problem is, on the one hand, a generalization of the Bin Packing Problem (BPP), and, on the other hand, a particular case of the Project Scheduling Problem with multidisciplinary jobs and one limited non-accumulative resource. Earlier, we proposed a polynomial algorithm that constructs the optimal package for a given order of non-increasing BCs. This result generalizes a similar result for BPP. For Two-Bar Charts Packing Problem (2-BCPP), when each BC consists of two bars, the algorithm we have proposed constructs a package in polynomial time, the length of which does not exceed $2\ OPT+1$, where $OPT$ is the minimum possible length of the packing. As far as we know, this is the first guaranteed estimate for 2-BCPP. We also conducted a numerical experiment in which we compared the solutions built by our approximate algorithms with the optimal solutions built by the CPLEX package. The experimental results confirmed the high efficiency of the developed algorithms.


\end{abstract}

\section{Introduction}
When solving the problem of optimizing the investment portfolio in the oil and gas sector, we faced the following problem \cite{Erzin20}. Suppose that the territory of the oil and gas field divided into clusters. For each cluster, a set of projects for its development is known. The project is characterized, in particular, by annual oil production. If we know the start year of the project, then we know the volumes of oil production in the first and all subsequent years of the project. The production schedule for each project can be represented in the form of a bar chart (BC), in which the bar's height corresponds to the volume of production in the corresponding year. It is required to determine the year of the launch of each project in such a way that the execution time of all projects is minimal, and the annual production volume from all cluster deposits does not exceed a predetermined value $D$, which is defined, for example, by the pipeline capacity.

Imagine a horizontal strip of height $D$. Then, the problem described above comes down to finding packing of BCs in a part of the strip (a rectangle of height $D$) of minimum length. Moreover, when packing each BC, the bars corresponding to production volumes in different years can move vertically, but they are inextricably horizontal and cannot be rearranged. Fig. 1 shows an example of feasible packing of three BCs, from which it follows that projects (a) and (b) start in the first year, project (c) starts in the fourth year, and all projects end in year 5. That is, the length of the strip (rectangle) into which all BCs are packed is 5.

\begin{figure}[t]
\includegraphics[width=\textwidth]{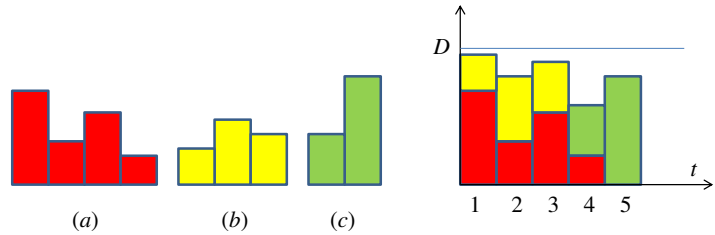}
\caption{Feasible packing of BCs.} \label{fig1}
\end{figure}

We were unable to find any publications on BCs packing. Similar problems that were studied reasonably well are the Bin Packing Problem (BPP) \cite{Baker85,Dosa07,Johnson73,Johnson85,Li97,Yue91,Yue95} and the problem of packing rectangles in a strip \cite{Baker80,Coffman80,Harren09,Harren14,Schiermeyer94,Steinberg97}.

In the classical BPP, a set $L$ of items, a size of each items, and a set of identical $D$-size containers (bins) are specified. All items must be placed in a minimum number of bins. One of the well-known algorithms for packing items in containers is First Fit Decreasing (FFD). As part of this algorithm, objects are numbered in non-increasing order. Then, all items are scanned in order, and the current item is placed in the first suitable bin.  In 1973, Johnson proved that the FFD algorithm uses no more than $11/9\ OPT(L)+4$ containers \cite{Johnson73}. In 1985 Backer showed that the additive constant can be reduced to 3 \cite{Baker85}. Yue in 1991 proved that $FFD(L)\leq 11/9\ OPT (L)+1$ \cite{Yue91}. Furthermore, in 1997 he improved the result to $FFD(L)\leq 11/9\ OPT (L)+7/9$ together with Li \cite{Li97}. In 2007, D\'{o}sa found the tight boundary of the additive constant and gave an example when $FFD(L)=11/9\ OPT(L)+6/9$ \cite{Dosa07}. A Modified First Fit Decreasing (MFFD) algorithm improves FFD by dividing items into groups by size and packing items from different groups separately. Johnson and Garey proposed this modification, and in 1985 they showed that $MFFD(L)\leq 71/60\ OPT(L)+31/6$ \cite{Johnson85}. Subsequently, the result was improved by Yue and Zhang to $MFFD(L)\leq 71/60\ OPT(L)+1$ \cite{Yue95}.

The BPP is a particular case of the problem of packing rectangles in a horizontal strip when all objects have the same width. Therefore, the problem of packing rectangles in a strip is NP-hard too. Moreover, if $P\neq NP$, then both problems are  3/2-inapproximable \cite{Vazirani01}. Formally, the problem of tight packing of rectangles into a semi-infinite strip of height $D$ is as follows. For each rectangle $i\in L$, we know the width $w_i$ and the height $h_i$. It is required to find the packing of the set of rectangles $L$ in a strip of minimum length. Rotation of rectangles is prohibited. The Bottom-Left algorithm proposed by Baker \cite{Baker80} arranges rectangles in descending order of height and has a ratio of 3. Coffman et al. \cite{Coffman80} in 1980 proposed algorithms with ratio 3 and 2.7. Sleator \cite{Sleator80} showed that his algorithm packs the rectangles into a strip whose length does not exceed $2\ OPT(L)+w_{max}(L)/2$, where $w_{max} (L)$ is the width of the widest rectangle in the set. Since $w_{max} (L) \leq OPT(L)$, the algorithm guarantees a ratio of 2.5. This ratio was reduced by Schiermeyer \cite{Schiermeyer94} and Steinberg \cite{Steinberg97} to 2. Harren and van Stee \cite{Harren09} were the first who evaluated a ratio of less than 2. Their proposed algorithm has a ratio of 1.9396. The smallest estimate for the ratio known to date obtained by Harren et al. \cite{Harren14} in 2014, and equals $(5/3+\varepsilon)OPT(L)$, for any $\varepsilon > 0$.

The problem under consideration is also a particular case of the project scheduling problem. Realy, each project consists of a sequence of jobs that needs to be done one after another without delay (no-wait). Each job has a unit duration and consumes the non-accumulative resource. It is required to find the start moment for each project in such a way that all projects are finished in minimum time, consuming together at most $D$ resource during each moment.

Resource-limited scheduling has been the subject of many publications discussing renewable and nonrenewable or accumulative resources. An overview of the results can be found, for example, in \cite{Hartmann02,Kolisch06}. For the case of an accumulative resource, exact and asymptotically exact algorithms have been developed \cite{Gimadi03,Gimadi19}. In the case of a limited renewable resource, the scheduling problem is NP-hard, and polynomial algorithms with guaranteed accuracy estimates are not known for it. As a rule, heuristic algorithms are developed for its approximate solution, and a posteriori analysis is performed \cite{Goncharov12,Goncharov14,Hartmann02,Kolisch06}. For example, in \cite{Gimadi19,Goncharov12,Goncharov14,Goncharov17,Goncharov19} multidisciplinary partially ordered jobs of arbitrary duration are considered that consume different amounts of a homogeneous resource at different time moments, and the authors developed approximate algorithms for solving the problem, and also performed a posteriori analysis, which showed a quite high efficiency. For comparison, the authors used a problem library PSPLIB \cite{PSPLIB96}. For some instances from the dataset J60 \cite{Goncharov12,Goncharov14,Goncharov17}, and dataset J120 \cite{Goncharov19}, the best-known solutions were improved. We were unable to find publications in which polynomial algorithms with guaranteed accuracy estimates proposed for such kind of project scheduling problem.

The problem of packing BCs in the particular case when all BCs consist of one bar is a BPP. In this article, to construct an approximate solution to the problem, we developed a greedy algorithm ($GA$) and obtained some qualitative results for non-increasing BCs with an arbitrary number of bars. The main result of the article is the proof that, for arbitrary two-bar charts, algorithm $A$, which uses the $GA$ as a procedure, builds a package whose length does not exceed $2\ OPT+1$, where $OPT$ is the minimum packing length. As far as we know, this is the first a priori estimate for the problem under consideration, which proves, in particular, that it belongs to the APX class.

The rest of the article is organized as follows. Section 2 provides a statement of the packing problem for BCs with an arbitrary number of bars, as well as a statement in the case of two-bar charts (2-BCs) in the form of Boolean Linear Programming (BLP). Section 3 describes the greedy algorithm $GA$ for the densest packing of BCs in a unit-height strip. Some properties of the algorithm are also given there. Section 4 discusses the packing problem of 2-BCs. We proposed an approximate algorithm that uses the $GA$ algorithm as a procedure. At the preliminary stage of the algorithm, some BCs are combined into one BC, in which at least one bar has a height greater than 1/2. This section gives the main result of the article, which consists in proving that the developed algorithm builds a package whose length does not exceed $2\ OPT + 1$, where $OPT$ is the minimum length of a strip into which all 2-BCs can be packed. In Section 5, we described the results of a numerical experiment, which made it possible to conduct a posteriori analysis of the developed algorithm. To build the optimal solution to the BLP, we used the CPLEX package and compared the optimal solution with the solution built by our algorithms. In section 6, we summarize and outline directions for further research.

\section{Formulation of the problem}
We have a horizontal stripe, the height of which, without loss of generality, equals 1 ($D=1$), and the set $S$ of BCs. Each BC $i\in S$ consists of the sequence of $l_i$ bars ($\max\limits_{i\in S}l_i=l$). Each bar $j$ in BC $i$ has a width that equals 1 and a height that equals $h^j_i\in (0,1]$ ($H_i=\max\limits_{j=1,\ldots,l_i}h^j_i$). We introduce the Cartesian coordinate system so that the lower boundary of the strip coincides with the abscissa. Let us consider a part of the strip to the right of the origin 0, which we divide into the \emph{cells} of width that equals 1, and number these cells with integers $1,2,\ldots$.

\begin{definition}
BC $i$ is \emph{non-increasing} (\emph{non-decreasing}) if $h^j_i\geq h^{j+1}_i$ ($h^j_i\leq h^{j+1}_i$) for all $j=1,\ldots,l_i-1$.
\end{definition}

\begin{definition}
\emph{Packing} is a function $p:S\rightarrow Z^+$, which associates with each BC $i$ an integer $p(i)$ corresponding to the cell number of the strip into which the first bar of BC $i$ falls.
\end{definition}

As a result of packing $p$, bars from BC $i$ occupy the cells $p(i),p(i)+1,\ldots,p(i)+l_i-1$.

\begin{definition}
The packing is \emph{feasible} if the sum of the heights of the bars that fall into one cell of the strip does not exceed 1. That is for each cell $k$ the inequality
$$
 \sum\limits_{i\in S:p(i)\leq k\leq p(i)+l_i-1} h^{k-p(i)+1}_i \leq 1
$$
holds.
\end{definition}

\begin{definition}
The packing \emph{length} $L(p)$ is the number of strip cells in which falls at least one bar.
\end{definition}

We can assume that any packing $p$ begins from the first cell, and in each cell from 1 to $L(p$), there is at least one bar. If this is not the case, then all or part of the package can be moved to the left.\\

\textbf{The Bar Charts Packing Problem (BCPP) is to build a feasible min-length package.}\\

Another measure of packing quality is \emph{density}. It is a ratio of the sum of the bar's heights to the packing length. The density cannot be greater than 1, and the higher the density, the better the packing.

Let us formulate the 2-BCPP problem, in which each BC has two bars, in the form of BLP. To do this, we introduce the following notation. Let the $i$th 2-BC have the height of the first bar $a_i$, and the second $b_i$. We introduce the variables:\\
$$
x_{ij}=\left\{
            \begin{array}{ll}
              1, & \hbox{if the first bar of BC $i$ is in the cell $j$;} \\
              0, & \hbox{else.}
            \end{array}
          \right.
$$
$$
y_j=\left\{
            \begin{array}{ll}
              1, & \hbox{if the cell $j$ contains at least one bar;} \\
              0, & \hbox{else.}
            \end{array}
          \right.
$$
Then 2-BCPP is written as follows.
\begin{equation}\label{e1}
  \sum\limits_j y_j \rightarrow\min\limits_{x_{ij},y_j\in\{0,1\}};
\end{equation}
\begin{equation}\label{e2}
  \sum\limits_j x_{ij} =1,\ i\in S;
\end{equation}
\begin{equation}\label{e3}
  \sum\limits_i a_ix_{ij} + \sum\limits_k b_kx_{k,j-1}\leq y_j,\ \forall j.
\end{equation}

Both problems BCPP and 2-BCPP are strongly NP-hard as the generalizations of the BPP \cite{Johnson73}. Moreover, these problems are $3/2$-inapproximable \cite{Vazirani01}.

\section{Greedy algorithm $G$}
First, we describe the version of the greedy algorithm (denote it by $G$), which builds an order-preserving package. That is, if the elements of the set $S$ are ordered, then the first bar of the $i$th BC cannot be placed to the right of the first bar of the $j$th BC if and only if $i<j$.

Let the elements of the set $S$ be arbitrarily numbered (ordered) by integers from 1 to $n=|S|$. We denote the resulting ordered set by $P$. In algorithm $G$, the first BC in $P$ is placed starting from the first cell and excluded from $P$. Then the following procedure is repeated for current list $P$. For the first item in the $P$, a cell is searched with the minimum number, not to the left of the cell containing the first bar of the previous BC, starting from which it can be placed preserving the feasibility of the packing. We exclude the first BC from the list $P$. The process continues until $P\neq\emptyset$.

Algorithm $G$ constructs a feasible package for a specific permutation of BCs with $O(nl)$ time complexity. In \cite{Erzin20}, we proved the lemma, which can be rephrased for BCPP as follows.

\begin{lemma}\cite{Erzin20}
 If each BC is non-increasing, then for a given order of BCs, algorithm $G$ constructs the optimal order-preserving solution to the BCPP.
\end{lemma}

The statement of the lemma generalizes the following statement for the BPP \cite{Lewis09}: There always exists at least one ordering of items that allows first-fit to produce an optimal solution.

Since the order does not matter for identical BCs, it is true the following

\begin{corollary}
 If all BCs are equal non-increasing or non-decreasing, then the algorithm $G$ constructs a package of minimum length.
\end{corollary}

In case of non-decreasing BCs, the optimal solution can be built by an algorithm similar to $G$ when the package is constructed from right to left. Since 2-BCs are either non-increasing or non-decreasing, the following also holds.

\begin{corollary}
 If all 2-BCs are equal, then the algorithm $G$ constructs an optimal solution to the 2-BCPP.
\end{corollary}

\section{2-BCPP}
In this section, we consider the problem in the case where each BC consists of two bars. To denote it, we use the 2-BCPP entry. Since 2-BCPP is a generalization of the BPP, let us try to use it. To do this, put each BC $i$ in a minimal 2-width rectangle that has a height equals $H_i = \max \{a_i, b_i\}$. As a result, we get a set of items $L$, each of which $i$ has a width equals 2 and characterized only by its height $H_i$. That is, we got BPP, for the solution of which we can use well-known approximate algorithms. Using the MFFD algorithm, it is possible to construct a package of items from $L$ using at most $71/60\ OPT(L)+1$ bins, where $OPT(L)$ is the minimum number of containers for packing items from $L$ \cite{Yue95}. This packing is feasible for 2-BCPP too. Since the height of the minimum bar in each rectangle containing BC can be arbitrarily small, the package constructed by the MFFD algorithm for 2-BCPP can have the density of 2 times less than the packing density of items from $L$. Therefore, the strip length, in which the MFFD algorithm packs all 2-BCs, is limited to $142/60\ OPT+2 \approx 2.367\ OPT+2$, where $OPT$ is the minimum packing length of 2-BCs.

\begin{remark}
Any BC $i\in S$ fits a rectangle of width $l_i$ and height $H_i$, and we can pack the resulting set of rectangles $R$ using the known algorithms. For example, an algorithm from \cite{Harren14} will construct (without rotation of the rectangles) a package whose length is not more than $(5/3 + \varepsilon)\ OPT(R)$, where $OPT(R)$ is the minimum possible packing length of the rectangles of the set $R$, and $\varepsilon> 0$. The resulting solution will be valid for BCPP too and will have a length at most $l (5/3+\varepsilon)OPT$, where $OPT$ is the optimal value of the BCPP objective function.

If all BCs have the same width $l$, then the MFFD algorithm for BPP will build a package for BCPP with a length of no more than
$l(71/60\ OPT+1)$.
\end{remark}

Below for 2-BCPP we propose a greedy algorithm, which is somewhat different from $G$. We denote it by $GA$. Let, as before, the list $P$ be an ordered set of elements from $S$. The first element in $P$ is placed in cells 1 and 2 and removed from $P$. Let some BCs are packed and deleted from $P$. Items deleted from $P$ do not move further. Then the typical procedure is performed, which consists of the following.  For each BC in $P$, we search the leftmost position that does not violate the feasibility of the packing. Among BCs that could be placed to the left of all, choose BC with a minimum number, fix its position in the package and delete it from $P$. The algorithm stops when $P = \emptyset$.

Algorithm $G$ builds an order-preserving package. As a result of the algorithm $GA$, BCs with higher numbers can stand in the package to the left of BCs with lower numbers.  Depending on the order of BCs in $P$, the algorithm $GA$ builds different solutions. Further, we propose algorithm $A$, using the $GA$ as a procedure, which builds a package of at most $2\ OPT + 1$ length, where $OPT$ is the smallest possible packing length. Algorithm $A$ consists of three stages.

The first stage is preparatory, and it consists of combining BCs so that, if it is possible, in each 2-BC, the height of at least one bar is greater than $1/2$. For this, the pair of BCs $i$ and $j$, for which $a_i, b_i, a_j, b_j \leq 1/2$, are combined into one new BC of width 2 with the height of bars $a_i + a_j$ and $b_i + b_j$. As a result, the set $S$ is transformed: two BCs $i$ and $j$ are removed from it, and one new BC is added. The procedure of combining BCs is repeated until $S$ contains the pairs of BCs with both bars no more than $1/2$. As a result, for each BC, except possibly one, the maximum bar's height will be greater than $1/2$ (Fig. 2$b$).

\begin{figure}[t]
\includegraphics[width=\textwidth]{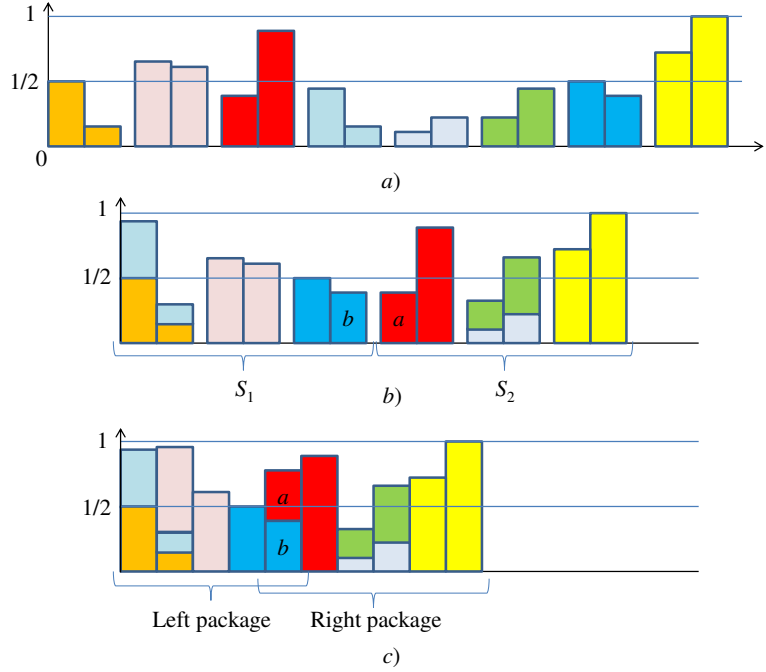}
\caption{Illustration of the operation of algorithm $A$. a) a set $S$; b) a set of combined BCs; c) a package constructed by algorithm $A$.} \label{fig2}
\end{figure}

At the second stage of algorithm $A$, we split the updated set $S$ into two subsets $S_1$ and $S_2$. In $S_1$ we include non-increasing, and in $S_2$ non-decreasing BCs (Fig. 2b). Without loss of generality, we assume that a small BC, for which both bars are at most $1/2$, if it exists, is the last element of the set $S_1$. We pack the set of elements in $S_1$ using the algorithm $GA$ (from left to right), and we pack the set of elements in $S_2$ by the analog of the algorithm $GA$ from right to left. Then we get two packages: left and right.

At the third stage of algorithm $A$, we shift the right package of the elements of the set $S_2$ to the left to the maximum so as not to violate the feasibility of the packing (Fig. 2c).

\begin{lemma}
 The time complexity of algorithm $A$ is $O(n^2)$.
\end{lemma}
\begin{proof}
The first stage of algorithm $A$ can be implemented with a time complexity of $O(n)$ as follows. Put $M=\emptyset$. Browse the BCs in numerical order. If the next BC has the first bar greater than the second and more than $1/2$, then we put it in the set $S_1$. If the next BC has a second bar greater than the first and more than $1/2$, then we put it in the set $S_2$. If both bars do not exceed $1/2$ and $M=\emptyset$, then put the current BC in $M$ and continue viewing. If both bars do not exceed $1/2$ and $M\neq\emptyset$, then combine the current BC with BC in $M$. If the resulting new combined BC has at least one bar greater than $1/2$, then exclude it from $M$ and put it in $S_1$ if the first bar is larger than the second; otherwise we put it in $S_2$. If, after combining the current BC with BC in $M$, a new BC with both bars not exceeding $1/2$ is obtained, then we leave the new merged BC in $M$ and proceed to consider the next BC in $P$. Thus, as a result of one scan of all BCs, we will form the sets $S_1$ and $S_2$.

At the second stage, we pack the elements of the sets $S_1$ and $S_2$ separately using the algorithm $GA$. At each step of the algorithm $GA$, one BC is added to the already constructed package. The complexity of the process of finding the best position for the current BC equals $O(n)$. The number of steps of the algorithm $GA$ is at most $n$. Therefore, the time complexity of the second stage of algorithm $A$ is $O(n^2)$.

The complexity of the third stage is $O(1)$. The lemma is proved.
\end{proof}

\begin{theorem}
Algorithm $A$ with time complexity $O(n^2)$ constructs a package for 2-BCPP whose length is at most $2\ OPT+1$, where $OPT$ is the minimum length of a strip into which all 2-BCs can be packed.
\end{theorem}
\begin{proof}
We first consider the packing of the set $S_1$. In this set, all BCs are not increasing, and the height of the first bar for all BCs, except, perhaps, the last BC, is greater than $1/2$. Algorithm $GA$, when packing next BC, shifts it as far as possible to the left. Let $k$ BCs are already packed, then the height of the last bar is $b_k$. Consider the $(k+1)$th BC. The following two cases are possible:
\begin{enumerate}
  \item $a_{k+1} + b_k \leq 1$;
  \item $a_{k+1} + b_k > 1$.
\end{enumerate}
In the first case, we place the first bar of the $(k+1)$th BC over the second bar of the $k$th BC. In the second case, we put the $(k+1)$th BC in the next two free cells. In both cases, after adding the $(k+1)$th BC, the total packing density of the first non-empty cells, except, possibly, the last, is greater than $1/2$.

For the set $S_2$, we carry out packing in a similar way from right to left. As a result, the packing density of non-empty cells without the first cell of the right package will be more than $1/2$.

Let us denote by $b$ the height of the last bar in the left package, and by $a$ the height of the first bar in the right package. After shifting the right package to the left to the maximum so that the total packing is feasible, we get one of the following cases.
\begin{itemize}
  \item $a+b>1$. Then the left and right packages are touching each other, and the density of the whole package is at least 1/2.
  \item $a+b\leq1$. Then the right bar of the left package and the left bar of the right package occupy the same cell (Fig. 2$c$). In this case, $a+b$ can be at least 1/2 or at most 1/2, and then the density of the whole package is at least 1/2, excluding maybe one cell.
\end{itemize}
In any way, we have that the package density of all cells, except, possibly, one, is more than $1/2$. From this, we obtain the statement of the theorem.
\end{proof}

\begin{remark}
Algorithm $A$ constructs a package whose density is below bounded by $1/2$, and this estimate is tight, which follows from the following instance. Let all BCs be equal with the height of the first bar $a_i = 1$, and the second $b_i = \varepsilon$, $i\in S$. Then each BC in the optimal package will occupy two cells, and the density of such packing $(1+\varepsilon)/2$ tends to $1/2$ when $\varepsilon$ tends to $0$. However, if instead of density, we compare the length of package constructed by the algorithm $GA$ with the minimum packing length, the difference will be less than two times. For the considered instance, for example, the ratio is $1$, i.e., $GA$ constructs the optimal solution. Therefore, to obtain a more accurate estimate for the ratio, one needs to find a more accurate lower bound for the length of the optimal packing.
\end{remark}

\begin{remark}
After the first step of algorithm $A$, several BCs are combined into one BC, which reduces the number of feasible packages. In addition, the length of the package constructed by the algorithm $GA$ substantially depends on the order of elements in the list $P$. In the next section, we present the results of a numerical experiment in which the lengths of the package constructed by the algorithm $GA$ for different BCs ordering are compared with the minimal packing length.
\end{remark}

\section{Simulation}
For simulation, all the proposed algorithms were implemented in the Python programming language. In the numerical experiment, the input data waer generated randomly. As parameters $a_i$, $b_i$ ($i=1,\ldots,n$) of the problem (1)-(3), random independent values were uniformly distributed in the segment $(0,1]$. We treated the instances of different sizes $n\in [10,1000]$. For each value of $n$, 100 different instances were generated. To build the optimal solution to BLP or to find the lower bound for the objective function, we used the IBM ILOG CPLEX 12.10 software package. The calculations were carried out on a computer Intel Core i7-3770 3.40GHz 16Gb RAM.

We examined six different approximate algorithms: the algorithms $A$, $GA$, and their modifications. To evaluate the influence of the first stage of algorithm $A$, we use algorithm $A1$ without the first stage. The quality of the solution built by algorithm depends significantly on the order of elements in the set $P$. To evaluate the influence of the ordering, we implemented the algorithms $A$\_LO, $GA$\_LO, and $A1$\_LO. The abbreviation ``LO'' means that before packing, we order the BCs lexicographically in non-increasing order of bar's height and then apply the algorithms $A$, $GA$, and $A1$, correspondently.

Table 1 presents the results of a numerical experiment. One can see the benefits of the preliminary lexicographic ordering (LO) of the BCs. CPLEX operating time was limited to 20 seconds when $n<500$, 40 seconds when $n=500$, 120 seconds when $n=750$ and 300 seconds when $n=1000$. For each size and each algorithm, the table shows the mean values $R_{av}$ and standard deviations $R_{sd}$ of $R$ which is defined as follows. If we know the optimal solution, then $R$ is the ratio. If CPLEX failed to find an optimal solution, then $R$ is the objective function of an approximate solution divided by the lower bound for objective function yielded by CPLEX during the allotted time. For $n\leq 75$, CPLEX often builds only an approximate solution to the problem, which is tight enough (the average value of $R$ is about 1.11). However, when the size increases up to 1000, CPLEX builds the approximate solution significantly worse than the package build by the proposed approximate algorithms. We show the graphics of $R$ depending on $n$ for the algorithms CPLEX, $A$, and $GA$\_LO in Fig. 3. In this figure, we also showed the standard deviation from the mean values of $R$ for different algorithms. It is also important to note that for all $n$, the running time of approximate algorithms did not exceed 1 second. For example, when $n=1000$, the algorithm $A$ built solutions in 0.25 seconds.

\begin{table}[h!]
\scriptsize{
\centering
\setlength{\tabcolsep}{1pt}
{
\begin{tabular}{||c||c|c||c|c||c|c||c|c||c|c||c|c||c|c||}
\hline

    \multirow{2}{*}{$n$} & \multicolumn{2}{|c|}{CPLEX} & \multicolumn{2}{|c|}{$A$} & \multicolumn{2}{|c|}{$A$\_LO} & \multicolumn{2}{|c|}{$A1$}
     & \multicolumn{2}{|c|}{$A1$\_LO} & \multicolumn{2}{|c|}{$GA$} & \multicolumn{2}{|c|}{$GA$\_LO}\\

\cline{2-15}
 & $R_{av}$ & $R_{sd}$& $R_{av}$ & $R_{sd}$  & $R_{av}$ & $R_{sd}$  & $R_{av}$ & $R_{sd}$  & $R_{av}$ & $R_{sd}$  & $R_{av}$ & $R_{sd}$  & $R_{av}$ & $R_{sd}$ \\

\hline

10	&1	&0.02	&1.21	&0.09	&1.16	&0.08	&1.2	&0.1	&1.12	&0.08	&1.14	&0.08	&1.07	&0.07 \\
25	&1.06	&0.04	&1.29	&0.09	&1.22	&0.08	&1.25	&0.07	&1.16	&0.06	&1.18	&0.06	&1.1	&0.05 \\
50	&1.08	&0.03	&1.28	&0.05	&1.21	&0.04	&1.22	&0.05	&1.13	&0.04	&1.16	&0.03	&1.09	&0.03\\
75	&1.11	&0.01	&1.27	&0.03	&1.2	&0.03	&1.21	&0.03	&1.12	&0.03	&1.15	&0.02	&1.08	&0.02\\
100	&1.11	&0.01	&1.26	&0.02	&1.19	&0.02	&1.19	&0.02	&1.11	&0.02	&1.14	&0.02	&1.08	&0.02\\
250	&1.15	&0.01	&1.22	&0.02	&1.17	&0.02	&1.15	&0.01	&1.07	&0.02	&1.11	&0.01	&1.05	&0.01\\
500	&1.18	&0.085	&1.19	&0.012	&1.16	&0.011	&1.12	&0.011	&1.05	&0.012	&1.09	&0.009	&1.04	&0.009\\
750	&1.24	&0.1	&1.18	&0.007	&1.15	&0.009	&1.11	&0.008	&1.04	&0.009	&1.08	&0.008	&1.03	&0.009\\
1000	&1.24	&0.076	&1.18	&0.008	&1.15	&0.009	&1.1	&0.008	&1.04	&0.009	&1.08	&0.007	&1.02	&0.007\\
\hline

\hline
\end{tabular}
\caption{Simulation results: mean values $R_{av}$ and standard deviations $R_{sd}$ of $R$.}
\label{table:1}
}}
\end{table}

\begin{figure}[t]
\includegraphics[width=\textwidth]{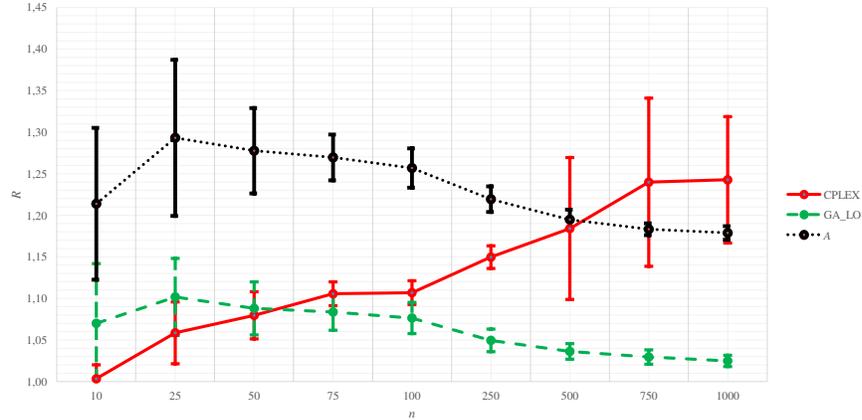}
\caption{Dependence of $R$ on the dimension $n$.} \label{fig3}
\end{figure}

Thus we can conclude that the algorithm $GA$\_LO (with lexicographic order of the BCs) is the best among all considered algorithms and, starting from $n=75$, it builds solutions more accurately than CPLEX in 5 minutes. The algorithm $A1$ turned out to be better in most cases than the algorithm $A$. We would also like to note that this experiment confirmed the significant influence of the ordering of BCs. All algorithms with preliminary lexicographic ordering turned out to be significantly more accurate than algorithms without preliminary ordering.

\section{Conclusion}
We examined the problem BCPP of packing BCs in a strip of minimum length. For the particular case, when all BCs have two bars each, the polynomial algorithm $A$ developed by us builds a package whose length does not exceed $2\ OPT+1$, where $OPT$ is the minimum possible package length. As far as we know, this is the first guaranteed estimate for 2-BCPP, which proves, in particular, that it belongs to the APX class. We also conducted a numerical experiment in which we compared the solutions built by our approximate algorithms with the optimal solutions built by the CPLEX package. Based on the results of a numerical experiment, we conclude that the algorithm $GA$\_LO, which uses the greedy algorithm $GA$ to pack BCs lexicographically ordered in non-increasing order, significantly outperforms all the others. In particular, for the number of BCs $n=1000$, it constructs a solution in less than 1 second, and the value of the objective function on this solution differs from the optimal value of the objective function by no more than 1.05 times (an average of 1.02 times). CPLEX in 5 minutes builds a solution on which the value of the objective function is, on average, 1.24 times worse than the optimal value of the objective function. On a larger dimension, CPLEX in 5 minutes does not produce a single feasible solution. If we increase the CPLEX operating time, the general trend will not change (see. Fig. 3).

In the future, we plan, firstly, to reduce the guaranteed estimate for the ratio, and, secondly, to consider BCs with a large number of bars.


\begin{thebibliography}{8}
\bibitem{Baker80}
Baker B.S., Coffman Jr. E.G., Rivest R.L.: Orthogonal packing in two dimensions. SIAM J. Comput. \textbf{9}(4), 846--855 (1980)

\bibitem{Baker85}
Baker B.S.:  A New Proof for the First-Fit Decreasing Bin-Packing Algorithm. J. Algorithms \textbf{6}, 49--70  (1985).

\bibitem{Coffman80}  Coffman Jr. E.G., Garey M.R., Johnson D.S., Tarjan R.E.: Performance bounds for level-oriented two-dimensional packing algorithms. SIAM J. Comput. \textbf{9}(4), 808--826 (1980)

\bibitem{Dosa07} D\'{o}sa Gy.: The Tight Bound of First Fit Decreasing Bin-Packing Algorithm Is $FFD(I)\leq 11/9\ OPT(I)+6/9$. Lecture Notes in Computer Sciences \textbf{4614}, 1--11 (2007)

\bibitem{Erzin20} Erzin A., et al.: Optimal Investment in the Development of Oil and Gas Field. CCIS (2020) (in the press)

\bibitem{Gimadi03} Gimadi E., Sevastianov S.: On Solvability of the Project Scheduling Problem with Accumulative Resources of an Arbitrary Sign. Selected papers in Operations Research Proceedings 2002. Berlin-Heidelberg: Springer Verlag, 241--246 (2003)

\bibitem{Gimadi19} Gimadi E.Kh., Goncharov E.N., Mishin D.V. : On Some Realizations of Solving the Resource Constrained Project Scheduling Problems. Yugoslav J. of Operations Research \textbf{29}(1), 31--42 (2019)

\bibitem{Goncharov12} Goncharov E.: A greedy heuristic approach for the Resource-Constrained Project Scheduling Problem. Studia Informatica Universalis. \textbf{9}(3), 79--90 (2011)

\bibitem{Goncharov14} Goncharov E.: A stochastic greedy algorithm for the resource-constrained project scheduling problem. Discrete Analysis and Operations Research. \textbf{21}(3), 11--24 (2014)

\bibitem{Goncharov17} Goncharov E.N., Leonov V.V.:  Genetic Algorithm for the Resource-Constrained Project Scheduling Problem. Automation and Remote Control. \textbf{78}(6), 1101--1114 (2017)

\bibitem{Goncharov19} Goncharov E.N. Variable Neighborhood Search for the Resource Constrained Project Scheduling Problem. Bykadorov et al. (Eds.): MOTOR 2019. CCIS \textbf{1090}, 39--50 (2019)

\bibitem{Harren09} Harren R., van Stee R.: Improved absolute approximation ratios for two-dimensional packing problems. In APPROX: 12th Int. Workshop on Approximation Algorithms for Combinatorial Optimization Problems, 177--189 (2009)

\bibitem{Harren14} Harren R., Jansen K., Pradel L., van Stee R.: A (5/3 + epsilon)-approximation for strip packing. Computational Geometry. \textbf{47}(2), 248--267 (2014)

\bibitem{Hartmann02} Hartmann S.: A self-adapting genetic algorithm for project scheduling under resource constraints. Naval Res. Logist. \textbf{49}, 433--448 (2002)

\bibitem{Johnson73} Johnson D.S.: Near-optimal bin packing algorithms. Massachusetts Institute of Technology. PhD thesis (1973)

\bibitem{Johnson85} Johnson D.S., Garey M.R.: A 71/60 theorem for bin packing. J. of Complexity. \textbf{1}(1), 65--106 (1985)

\bibitem{PSPLIB96} Kolisch R., Sprecher A.: PSPLIB -- a project scheduling problem library. Eur. J. Oper. Res. \textbf{96}(1), 205--216 (1996)

\bibitem{Kolisch06} Kolisch R., Hartmann S.: Experimental investigation of heuristics for resource-constrained project scheduling: an update. Eur. J. Oper. Res. \textbf{174}, 23--37 (2006)

\bibitem{Lewis09} Lewis R.: A General-Purpose Hill-Climbing Method for Order Independent Minimum Grouping Problems: A Case Study in Graph Colouring and Bin Packing. Computers and Operations Research. \textbf{36}(7), 2295--2310 (2009)

\bibitem{Li97} Li R., Yue M. The proof of $FFD(L)\leq 11/9\ OPT(L)+7/9$. Chinese Science Bulletin. \textbf{42}(15), 1262--1265 (1997)

\bibitem{Schiermeyer94}  Schiermeyer I.: Reverse-fit: A 2-optimal algorithm for packing rectangles. In ESA: Proc. 2nd European Symposium on Algorithms, 290--299 (1994)

\bibitem{Sleator80} Sleator D.D.: A 2.5 times optimal algorithm for packing in two dimensions. Inf. Process. Lett. \textbf{10}(1), 37--40 (1980)

\bibitem{Steinberg97} Steinberg A.: A strip-packing algorithm with absolute performance bound 2. SIAM J. Comput. \textbf{26}(2) 401--409 (1997)

\bibitem{Vazirani01}  VaziraniV.V.: Approximation Algorithms. Springer Berlin Heidelberg. (2001)

\bibitem{Yue91} Yue M.: A simple proof of the inequality $FFD(L)\leq 11/9\ OPT(L)+1,\ \forall L$, for the FFD bin-packing algorithm. Acta Mathematicae Applicatae Sinica. \textbf{7}(4), 321--331 (1991)

\bibitem{Yue95} Yue M., Zhang L.: A simple proof of the inequality $MFFD(L)\leq 71/60\ OPT(L)+1,\ \forall L$, for the MFFD bin-packing algorithm. Acta Mathematicae Applicatae Sinica. \textbf{11}(3), 318--330 (1995)

\end{thebibliography}
\end{document}